\documentclass[11pt]{article}

\usepackage{amsmath}
\usepackage{amssymb}
\usepackage{amsthm}
\usepackage{graphicx}
\usepackage{tikz}
\usepackage{xcolor}
\usepackage[ruled]{algorithm2e}

\allowdisplaybreaks

\usepackage[margin=1.2in]{geometry}
\usepackage{enumitem}
\usepackage{bookmark}

\newcommand*{\cM}{\mathcal{M}}
\newcommand*{\bI}{\mathbf{I}}
\newcommand*{\bJ}{\mathbf{J}}
\newcommand*{\bK}{\mathbf{K}}
\newcommand*{\bU}{\mathbf{U}}
\newcommand*{\bX}{\mathbf{X}}

\newcommand*{\EE}{\mathbb{E}}

\def\eps{\varepsilon}
\def\del{\delta}

\def\bI{\mathbf{I}}

\newcommand*{\Tmix}{T_{\mathrm{mix}}}

\def\cG{\mathcal {G}}

\def\cI{\mathcal {I}}

\def\1{\mathbf{1}}

\def\lam {\lambda}

\def\tce{t_c + \eps}
\def\tce2{t_c + \frac{\eps}{2}}

\def\bX{\mathbf{X}}

\newcommand*{\sharpP}{\#{}P}

\DeclareMathOperator*{\var}{var}
\DeclareMathOperator*{\cov}{cov}

\newtheorem*{theorem*}{Theorem}
\newtheorem{theorem}{Theorem}
\newtheorem{lemma}[theorem]{Lemma}
\newtheorem{cor}[theorem]{Corollary}

\newtheorem*{defn*}{Definition}

\newtheorem*{prop*}{Proposition}
\newtheorem{conj}[theorem]{Conjecture}
\newtheorem*{conj*}{Conjecture}

\newtheorem{question}[theorem]{Question}
\newtheorem*{fact*}{Fact}

\begin{document}
\title{Approximately Counting Independent Sets of a Given Size in Bounded-Degree Graphs}

\author{Ewan Davies\footnote{Department of Computer Science, University of Colorado Boulder. Corresponding address: Department of Computer Science, Colorado State University.} \and Will Perkins\footnote{Department of Mathematics, Statistics, and Computer Science, University of Illinois Chicago. Supported in part by NSF grants DMS-1847451 and CCF-1934915}}

\date{\today}

\maketitle
\begin{abstract}
    We determine the computational complexity of approximately counting and sampling independent sets of a given size in bounded-degree graphs. 
    That is, we identify a critical density $\alpha_c(\Delta)$ and provide (i) for $\alpha < \alpha_c(\Delta)$ randomized polynomial-time algorithms for approximately sampling and counting independent sets of given size at most $\alpha n$ in $n$-vertex graphs of maximum degree $\Delta$; and (ii) a proof that unless NP=RP, no such algorithms exist for $\alpha>\alpha_c(\Delta)$. 
    The critical density is the occupancy fraction of the hard core model on the complete graph $K_{\Delta+1}$ at the uniqueness threshold on the infinite $\Delta$-regular tree, giving $\alpha_c(\Delta)\sim\frac{e}{1+e}\frac{1}{\Delta}$ as $\Delta\to\infty$. 
    Our methods apply more generally to anti-ferromagnetic 2-spin systems and motivate   new questions in extremal combinatorics.
\end{abstract}

\section{Introduction}

Counting and sampling independent sets in graphs are fundamental computational problems arising in several fields including algorithms, statistical physics, and combinatorics. 
Given a graph $G$, let $\cI(G)$ denote the set of independent sets of $G$. 
The independence polynomial of $G$ is 
\[ Z_G(\lam) = \sum_{I \in \cI(G)} \lam ^{|I|}  = \sum_{k \ge 0} i_k(G) \lam^k  \,, \]
where $i_k(G)$ is the number of independent sets of size $k$ in $G$. 
The independence polynomial also arises as the partition function of the hard-core model from statistical physics.

With $G$ and $\lam$ as inputs, exact computation of $Z_G(\lam)$ is \sharpP-hard~\cite{Val79,Gre00}, but the complexity of approximating $Z_G(\lam)$ has been a major topic in recent theoretical computer science research. 
There is a detailed understanding of the complexity of approximating $Z_G(\lam)$ for the class of graphs of maximum degree $\Delta$, in particular showing that there is a \emph{computational threshold} which coincides with a certain probabilistic phase transition as one varies the value of $\lam$.

The hard-core model on $G$ at fugacity $\lam$ is the probability distribution on $\cI(G)$ defined by
\[ \mu_{G,\lam} (I) = \frac{\lam^{|I|}}{Z_G(\lam)} \,.\]
Defined on a lattice like $\mathbb Z^d$ (through an appropriate limiting procedure), this is a simple model of a gas (the hard-core lattice gas) and it exhibits an order/disorder phase transition as $\lam$ changes.  
The hard-core model can also be defined on the infinite $\Delta$-regular tree (the \emph{Bethe lattice}). Kelly~\cite{Kel85} determined the critical threshold for uniqueness of the infinite volume measure on the tree, namely
\begin{equation}\label{eq:lamc}
    \lam_c(\Delta) = \frac{(\Delta-1)^{\Delta-1}}{(\Delta-2)^\Delta}\,.
\end{equation}
This value of $\lam$  also marks a computational threshold for the complexity of approximating $Z_G(\lam)$ on graphs of maximum degree $\Delta$. 
One can approximate $Z_G(\lam)$ up to a relative error of $\eps$ in time polynomial in $n$ and $1/\eps$ with several different methods, provided $G$ is of maximum degree $\Delta$ and $\lam<\lam_c(\Delta)$.
The first such algorithm is based on correlation decay on trees and is due to Weitz~\cite{Wei06}, but recently alternative algorithms based on polynomial interpolation~\cite{Bar16,PR17,PR19} and Markov chains~\cite{ALG20,CLV20,CLV20a} for this problem have also been given. 
Conversely, for $\lam>\lam_c(\Delta)$ a result of Sly and Sun~\cite{SS14} and Galanis, \v{S}tefankovi\v{c}, and  Vigoda~\cite{galanis2016inapproximability} (following Sly~\cite{Sly10}) states that unless NP=RP there is no polynomial-time algorithm for approximating $Z_G(\lam)$ on graphs of maximum degree $\Delta$. 
Counting and sampling are closely related, and by standard reduction techniques the same computational threshold holds for the problem of approximately sampling independent sets from the hard-core distribution. 

The hard-core model is an example of the \emph{grand canonical ensemble} from statistical physics, where one studies physical systems that can freely exchange particles and energy with a reservoir. 
Closely related is the \emph{canonical ensemble}, where one removes the reservoir and considers a system with a fixed number of particles. In the context of independent sets in graphs, this corresponds to the uniform distribution on independent sets of some given size $k$. 
Here the number $i_k(G)$ of independent sets of size $k$ in $G$ plays the role of the partition function.
In this paper we answer affirmatively the natural question of whether there is a similar computational threshold for the problem of approximating $i_k(G)$, and the related problem of sampling independent sets of size $k$ approximately uniformly.
Analogous to the critical fugacity in the hard-core model, we identify a critical density $\alpha_c(\Delta)$, and for $\alpha < \alpha_c(\Delta)$ we give a fully polynomial-time randomized approximation scheme (FPRAS, defined below) for counting independent sets of  size $k$ in $n$-vertex graphs of maximum degree $\Delta$, where $0\le k\le \alpha n$. 
We also show that unless NP=RP there is no such algorithm for $\alpha > \alpha_c(\Delta)$. 

In statistical physics the grand canonical ensemble and the canonical ensemble are known to be  equivalent in some respects under certain conditions, and the present authors, Jenssen, and Roberts~\cite{DJPR18a} used this idea to give a tight upper bound on $i_k(G)$ for large $k$ in large $\Delta$-regular graphs $G$ (see also~\cite{DJP20} for the case of small $k$). 
Here, the main idea in our proofs is also to exploit the equivalence of ensembles. For algorithms at subcritical densities we approximately sample independent sets from the hard-core model and show that with sufficiently high probability we get an independent set of the desired  size, distributed approximately uniformly. 
For hardness at supercritical densities we construct an auxiliary graph $G'$ such that $i_k(G')$ is approximately proportional to $Z_G(\lam)$ for some $\lam>\lam_c(\Delta)$, and hence is hard to approximate. 
Our counting and sampling algorithms for  independent sets of size $k$ permit higher densities than previous algorithms for this problem based on Markov chains~\cite{BD97,AL20}, and an algorithm implicit in~\cite{DJP20} based on the cluster expansion.

One feature of our methods is the incorporation of several advances from recent research on related topics. 
From the geometry of polynomials we use a state-of-the-art zero-free region for $Z_G(\lam)$ due to Peters and Regts~\cite{PR19} and a central limit theorem of Michelen and Sahasrabudhe~\cite{MS19a,MS19} (though an older result of Lebowitz, Pittel, Ruelle and Speer~\cite{LPRS16} would also suffice), and we also apply the very recent development that a natural Markov chain (the Glauber dynamics)  for sampling from the hard-core model  mixes rapidly at fugacities $\lam < \lam_c(\Delta)$ on all graphs of maximum degree at most $\Delta$~\cite{AL20,CLV20a}.  Finally, our results also show a connection between these algorithmic and complexity-theoretic problems and extremal combinatorics problems for bounded-degree graphs~\cite{CR14,DJPR18a,DJP20}, see also the survey~\cite{zhao2017extremal}.

Due to the combinatorial significance of independent sets, the case of the hard-core model is a natural starting point for our work establishing a computational threshold for the problem of approximating coefficients of  partition functions. 
There are other well-studied models for which it is interesting to consider an analogous computational problem, and we show that our methods also apply  to the case of anti-ferromagnetic 2-spin models. 
\subsection{Preliminaries}

Given an error parameter $\eps$ and real numbers $z$, $\hat z$, we say that $\hat z$ is a \emph{relative $\eps$-approximation} to $z$ if $e^{-\eps}\le \hat z/z \le e^\eps$. 
A \emph{fully polynomial-time randomized approximation scheme} or \emph{FPRAS} for a counting problem is a randomized algorithm that with probability at least $3/4$ outputs a relative $\eps$-approximation to the solution of the problem in time polynomial in the size of the input and $1/\eps$. 
If the algorithm is deterministic (i.e.\ succeeds with probability 1) then  it is a \emph{fully polynomial-time approximation scheme} (\emph{FPTAS}).
An $\eps$-approximate sampling algorithm for a probability distribution $\mu$ outputs a random sample from a distribution $\hat\mu$ such that the total variation distance $\Vert\mu-\hat\mu\Vert_{TV} \le \eps$, and an \emph{efficient sampling scheme} is, for all $\eps>0$ an $\eps$-approximate sampling algorithm which runs in time polynomial in the size of the input and $\log(1/\eps)$. 
Note that approximate sampling schemes whose running times are polynomial in $1/\eps$ or in $\log(1/\eps)$ are common in the literature, but we adopt the stronger definition  for this paper. 
The inputs to our algorithms are graphs, and input size corresponds to the number of vertices of the graph.

An \emph{independent set} in a graph $G=(V,E)$ is a subset $I\subset V$ such that no edge of $E$ is contained in $I$. 
The \emph{density} of such an independent set $I$ is $|I|/|V|$, and it will be convenient for us to parametrize independent sets by their density instead of their size. 
We write $\cI(G)$ for the set of all independent sets in $G$, $\cI_k(G)$ for the set of independent sets of size $k$ in $G$, and $i_k(G)=|\cI_k(G)|$ for the number of such sets. 
Recall the hard-core distribution $\mu_{G,\lam}$ on $\cI(G)$ is given by $\mu_{G,\lam}(I) = \lam^{|I|}/Z_G(\lam)$.
We also define the \emph{occupancy fraction} $\alpha_G(\lam)$ of the hard-core model on $G$ at fugacity $\lam$ to be the expected density of a random independent set drawn according to $\mu_{G,\lam}$.
A simple calculation gives that 
\[ \alpha_G(\lam) = \frac{1}{|V(G)|}\frac{\lam Z'_G(\lam)}{Z_G(\lam)}. \]
Let $\cG_\Delta$ be the set of $\Delta$-regular graphs and $\cG_{\le\Delta}$ be the set of graphs of maximum degree $\Delta$.

The critical density that we show constitutes a computational threshold for the problems of counting and sampling independent sets of a given size in graphs in $\cG_\Delta$ or $\cG_{\le\Delta}$ is
\[ \alpha_c(\Delta) = \frac{\lam_c(\Delta)}{1+(\Delta+1)\lam_c(\Delta)} = \frac{(\Delta-1)^{\Delta-1}}{(\Delta-2)^{\Delta}+(\Delta+1)(\Delta-1)^{\Delta-1}}\,, \]
with $\lam_c$ the critical fugacity as in~\eqref{eq:lamc}. This may seem unexpected at first sight, but has a natural interpretation. 
The threshold is in fact the quantity $\alpha_{K_{\Delta+1}}(\lam_c(\Delta))$, the occupancy fraction of the complete graph on $\Delta+1$ vertices at the critical fugacity $\lam_c(\Delta)$. 
This is a natural threshold because the occupancy fraction is a monotone increasing function of $\lam$, and the complete graph on $\Delta+1$ vertices has the minimum occupancy fraction over all graphs of maximum degree $\Delta$. 
Thus, the value of $\lam$ which makes $\alpha_G(\lam)>\alpha_c(\Delta)$ for all $G\in\cG_{\le \Delta}$ must be greater than $\lam_c(\Delta)$. Conversely, if $\alpha < \alpha_c(\Delta)$ then for every graph $G\in\cG_{\le \Delta}$ there is some $\lam<\lam_c(\Delta)$ such that $\alpha_G(\lam) = \alpha$.

\subsection{Main results}

We are now ready to state our main result.
\begin{theorem}\label{thm:main}\hfill
\begin{enumerate}[label=\textup{(\alph*)}]
    \item\label{itm:algs} For every $\alpha < \alpha_c(\Delta)$ there is an FPRAS for $i_{\lfloor \alpha n \rfloor}(G)$ and an efficient sampling scheme for the uniform distribution on $\cI_{\lfloor \alpha n \rfloor}(G)$ for $n$-vertex graphs $G$ of maximum degree $\Delta$.
    \item\label{itm:hardness} Unless \textup{NP=RP}, for every $\alpha  \in (\alpha_c(\Delta),1/2)$ there is no FPRAS for $i_{\lfloor \alpha n \rfloor}(G)$  for $n$-vertex, $\Delta$-regular graphs $G$.
\end{enumerate}
\end{theorem}
The assumption NP$\ne$RP, which is that randomized polynomial-time algorithms cannot solve NP-hard problems, is standard in computational complexity theory. Indeed, this assumption is used in~\cite{Sly10,SS14,galanis2016inapproximability} to show hardness of approximation for $Z_G(\lam)$ on regular graphs at supercritical fugacities, which we apply directly.  The upper bound of $1/2$ on $\alpha$ in~\ref{itm:hardness} is required since in a regular graph (of degree $\ge 1$) there are no independent sets of density greater than $1/2$ and counting those of density $1/2$ amounts to counting connected components in a bipartite graph (which can be done in polynomial time). For graphs of maximum degree $\Delta$ there is no such barrier, and in this case our methods can also be used to prove~\ref{itm:hardness} for $\alpha\in(\alpha_c(\Delta),1)$.

On the algorithmic side, Bubley and Dyer~\cite{BD97} showed via path coupling that a natural Markov chain for sampling independent sets of size $k$ in $n$-vertex graphs of maximum degree $\Delta$ mixes rapidly when $k< n/(2\Delta+2)$, and recently this was slightly improved to $k < n/(2\Delta)$ via the method of high-dimensional expanders by Alev and Lau~\cite{AL20} (who also gave an improved bound in terms of the smallest eigenvalue of the adjacency matrix of $G$). The fast mixing of this Markov chain provides a randomized algorithm for approximate sampling and an FPRAS for approximate counting for this range of $k$.
Implicit in the work of the present authors and Jenssen~\cite{DJP20} is an alternative method based on the cluster expansion that yields an FPTAS for  $i_k(G)$ when $k < e^{-5}n/(\Delta+1)$, and although we did not try to optimize the constant it seems unlikely that without significant extension the cluster expansion approach could yield a sharp result.
Considering asymptotics as $\Delta \to \infty$, these previous algorithms work for densities up to $(c+o(1))/\Delta$ with the constant $c$ being $1/2$ or $e^{-5}\approx 0.007$ respectively. 
Here, our algorithms work for densities $\alpha$ satisfying
\[ \alpha < \alpha_c(\Delta) = ( 1 +o(1)) \frac{e }{1+e}\frac{1}{\Delta}\,, \]
as $\Delta\to\infty$. 
The constant $e/(1+e)$ is approximately $0.731$, and our hardness proof shows that $\alpha_c(\Delta)$ is a tight threshold.

Our sampling algorithm is based on searching over possible values of $\lam$ until we find one for which the expected size of an independent set from the hard-core model is close to the target $k$.  We then resort to rejection sampling to obtain an independent set of the desired size: we repeatedly sample from the hard-core model and output the first independent set of size $k$ that is produced.  Our approximate counting algorithm is based on a standard reduction of approximate counting to approximate sampling.   

The method of sampling from the canonical ensemble by sampling from the grand canonical ensemble and conditioning on obtaining an object of the desired size is quite old and appears in the seminal papers of Jerrum and Sinclair~\cite{jerrum1989approximating,JS93}.  The key technical step in applying this method is to prove a lower bound on the probability of obtaining the desired size; in~\cite{jerrum1989approximating,JS93} this is accomplished by using log-concavity of  the specific distribution on sizes.  Since the hard-core model does not have this property in general, we need a different argument; our new argument is based on the rapid mixing of Glauber dynamics.

 Harris and Kolmogorov~\cite{HK20} have recently investigated the general  problem of estimating  the coefficients of  partition functions given access to samples from the corresponding Gibbs distribution.  Applying their ideas and results to this problem could likely lead to more efficient run times for the problem of approximating the entire sequence $ \{ i_k(G) \}$,  $ 1\le k \le \alpha n$.  Their ideas could also be applied to improve the efficiency of the reduction from counting to sampling, as they describe a more efficient `cooling schedule' for the simulated annealing technique that we use for this type of reduction.

As an additional application of our techniques we find an approximate computational threshold for the class of triangle-free graphs.
\begin{theorem}\label{thmTriangleFree}
For every $\delta >0$ there is $\Delta_0$ large enough so that the following is true.
\begin{enumerate}[label=\textup{(\alph*)}]
\item\label{itmTriangleFreeAlgs}
    For $\Delta \ge \Delta_0$ and $\alpha < \frac{1-\delta}{\Delta}$ there is an FPRAS and efficient sampling scheme for $i_{\lfloor \alpha n \rfloor}(G)$ for the class of triangle-free graphs of maximum degree $\Delta$.   
\item\label{itmTriangleFreeHardness}
    For $\Delta \ge \Delta_0$ and $\alpha \in \left(\frac{1+\delta}{\Delta},1/2\right)$ there is no FPRAS for $i_{\lfloor \alpha n \rfloor}(G)$ for the class of triangle-free graphs of maximum degree $\Delta$.  
\end{enumerate}
\end{theorem} 

\noindent
The proof of this theorem uses a result on the occupancy fraction of triangle-free graphs from~\cite{DJPR17a}, see Section~\ref{secTriangles}.

\subsection{Extensions to anti-ferromagnetic \texorpdfstring{$2$}{2}-spin systems}
\label{secIntroanti-ferro}

The hard-core model is one important example of an \emph{anti-ferromagnetic 2-spin system}.  Another example is the anti-ferromagnetic Ising model (with or without an external field).  Extensions of the results of Weitz and Sly establish the existence of a critical threshold (in the external field) for the efficient approximability of the partition function of an anti-ferromagnetic 2-spin system on a bounded degree graph~\cite{LLY13,SS14,SST14,GSV15}.  Given Theorem~\ref{thm:main}, one can ask if a similar phenomenon occurs for the problem of approximating individual coefficients of the partition function.  

We first specialize to the case of $\Delta$-regular graphs, in which case all anti-ferromagnetic 2-spin systems can be captured by either the hard-core model or the anti-ferromagnetic Ising model (see e.g.,~\cite{SST14}).   The  Ising model is defined by the partition function 
 \[ Z_G(B,\lam) := \sum_{\sigma : V(G)\to \{+,-\}}B^{m_G(\sigma)}\lam^{n_G(+,\sigma)}, \]
 where $\sigma$ is an assignment of a spin in $\{+,-\}$ to each vertex of $G$, $m_G(\sigma)$ is the number of monochromatic edges of $G$ under $\sigma$ (meaning that both of the endpoints get the same spin), and $n_G(+,\sigma) = |\sigma^{-1}(+)|$ is the number of vertices with spin $+$. 
The Ising model is anti-ferromagnetic if $B \in (0,1)$, and without loss of generality we may consider the case $0 < \lam  \le 1$ as the model is symmetric under the transformation $\lam \mapsto 1/\lam$. 
We define the measure $\mu_{G,B,\lam}$ on spin assignments of $G$ given by 
\[\mu_{G,B,\lam}(\sigma)= \frac{B^{m_G(\sigma)}\lam^{n_G(+,\sigma)}}{Z_G(B,\lam) }  \,. \]

The Ising model can be defined on the infinite $\Delta$-regular tree via the DLR equations.  There is a region in the space of parameters $(\Delta, B, \lam)$ such that the corresponding Gibbs measure is unique. 
For $\Delta \le 2$ the measure is unique for all $B,\lam$, so we assume $\Delta\ge 3$. 
Let $B_c(\Delta) = (\Delta-2)/\Delta$.   For $B_c < B \le 1$ the Gibbs measure is unique for all choices of $\lam$.
For $0 < B < B_c$ the uniqueness region is given by $0 < \lam \le \lam_c(B,\Delta)$ for some critical value $\lam_c(B,\Delta)$ (given by the solution to an explicit equation), and correspondingly we have non-uniqueness for $\lam_c < \lam \le 1$. 
See e.g.~\cite{Geo11} for a derivation of these parameter regions.
The results of Weitz~\cite{Wei06}, Sinclair, Srivastava, and Thurley~\cite{SST14}, Sly and Sun~\cite{SS14}, and Galanis, {\v S}tefankovi{\v c}, and Vigoda~\cite{GSV15} establish that for fixed $0 < B < B_c$, the value $\lam_c(B,\Delta)$ is a computational threshold for the problem of approximating $Z_G(B,\lam)$ on $\Delta$-regular graphs: there is an FPTAS for the uniqueness region $0 < \lam < \lam_c(B,\Delta)$, and unless NP=RP, there is no FPRAS in the region $\lam_c < \lam \le 1$.

We can generalize the problem we study for the hard-core model and consider the analogous problem of approximating coefficients of the Ising model partition function, viewed as a polynomial in $\lam$. Write
\[ Z_G(B, \lam) = \sum_{k=0}^n c_k \lam^k  \]
where the coefficients $c_k = c_k(G,B)$ depend on the graph $G$ and on $B$.  Write $\mathcal B_k(G)$ for the set of $+/-$ assignments to $V(G)$ with exactly $k$ $+$ spins, and let $\nu_{G,B,k}$ be the probability distribution on $\mathcal B_k(G)$ given by
\[  \nu_{G,B,k} (\sigma) = \frac{ B^{m_G(\sigma)}}{ \sum_{\sigma' \in \mathcal B_k(G)} B^{m_G(\sigma')}}    \, . \]
This is the measure $\mu_{G,B,\lam}$ conditioned on the event $\sigma \in \mathcal B_k(G)$.

Define the occupancy fraction $\alpha_G(B,\lam)$ as the expected fraction of $+$ spins in a sample from the distribution $\mu_{G,B,\lam}$.  In particular,
\[\alpha_G(B,\lam)= \frac{\lam}{|V(G)|}\frac{\partial}{\partial \lam} \log Z_G(B,\lam) \,.\]

We can now define the associated extremal problem.  Let 
\[ \alpha_{\mathrm{inf}}( B, \lam, \Delta) = \inf_{G \in \mathcal G_{\Delta}} \alpha_G(B,\lam) \]
be the solution to the extremal problem of minimizing the occupancy fraction over the class $\mathcal G_{\Delta}$ of all $\Delta$-regular graphs.  
Our result for the anti-ferromagnetic Ising model on regular graphs connects the solution to this extremal problem to a computational threshold for approximating $c_k(G,B)$ in $\Delta$-regular graphs.

\begin{theorem}\label{thmanti-ferro}
For $\Delta \ge 3$ and  $B< B_c(\Delta)$,  let $\lam_c = \lam_c (\Delta,B)$ be the uniqueness threshold for the anti-ferromagnetic Ising model on the infinite $\Delta$-regular tree.  Then
\begin{enumerate}[label=\textup{(\alph*)}]
    \item\label{itm:algsgen} For every $\alpha < \alpha_{\mathrm{inf}}(B, \lam_c,\Delta)$ there is an FPRAS for $c_{\lfloor \alpha n \rfloor}(G,B)$ and an efficient sampling scheme for $\nu_{G,B,\lfloor \alpha n \rfloor}$ for all $n$-vertex, $\Delta$-regular graphs $G$.
    \item\label{itm:hardnessgen} Unless \textup{NP=RP}, for every $\alpha \in (\alpha_{\mathrm{inf}}(B, \lam_c,\Delta),1/2]$ there is no FPRAS for $c_{\lfloor \alpha n \rfloor}(G,B)$ for $n$-vertex, $\Delta$-regular graphs $G$.
\end{enumerate}
\end{theorem}

By the remark above, Theorem~\ref{thm:main} and Theorem~\ref{thmanti-ferro} cover all non-trivial $2$-spin anti-ferromagnetic systems on $\Delta$-regular graphs.  The reason for the upper bound $\alpha\le 1/2$ in~\ref{itm:hardnessgen} is the same as for considering $\lam \le 1$; a result in the complementary range follows by swapping the roles of the spins. 

This theorem provides a direct connection between extremal problems for occupancy fractions  and thresholds for approximating coefficients of partition functions.  To obtain explicit computational thresholds one must solve the relevant extremal problem.  In Section~\ref{sec:further} below we pose a conjecture which would provide an explicit value for the threshold $\alpha_{\mathrm{inf}}(B, \lam_c,\Delta)$ in Theorem~\ref{thmanti-ferro}.

\subsection{Related work}

Counting independent sets of a specified size has arisen in various places as a natural fixed-parameter version of counting independent sets, and is equivalent to counting cliques of a specified size in the complement graph. 
Exact computation of $i_k(G)$ in an $n$-vertex graph $H$ is trivially possible in time $O(k^2n^k)$, though improvements can be made via fast matrix multiplication algorithms (see e.g.~\cite{EG04}). 
Another branch of research concerns the complexity (in both time and number of queries to the graph data structure) of counting and approximately counting cliques. 
For example, in~\cite{ERS20} the authors gave a randomized approximation algorithm for approximating the number of cliques of size $k$.
Results of this kind perform poorly in our setting, which is equivalent to counting cliques in the complement of bounded-degree graphs, because such graphs are very dense. 
In particular, the main result of~\cite{ERS20} has expected running time $\Omega((nk/e)^k)$ in our setting.

With a focus on bounded-degree graphs and connections to statistical physics, our work is closer in spirit to that of Curticapean, Dell, Fomin, Goldberg, and Lapinskas~\cite{curticapean2019fixed}. 
There, the authors consider the problem of counting independent sets of size $k$ in bipartite graphs from the perspective of parametrized complexity. 
They give algorithms for exact computation and approximation of $i_k(G)$ in bipartite graphs (of bounded degree and otherwise), including a fixed parameter tractable randomized approximation scheme, though their running times are exponential in $k$.
We note that the complexity of approximately counting the total number of independent sets in bipartite graphs (a problem known as \#{}BIS) is unknown~\cite{dyer2004relative}.

\subsection{Further questions}
\label{sec:further}

For the hard-core model, the algorithm of Weitz~\cite{Wei06} gives a deterministic approximation algorithm (FPTAS) for $Z_G(\lam)$ for $\lam < \lam_c(\Delta)$.  The approach of Barvinok along with results of Patel and Regts and Peters and Regts give another FPTAS for the same range of parameters~\cite{Bar16,PR17,PR19}.   Our algorithm for approximating the number of independent sets of a given size uses randomness, but we conjecture that there is a deterministic algorithm that works for the same range of parameters. (The cluster expansion approach of~\cite{DJP20} gives an FPTAS but only for smaller values of $\alpha$).
\begin{conj}\label{Qdeterministic}
There is an FPTAS for $i_{\lfloor \alpha n \rfloor}(G)$ for $G \in \cG_{\le\Delta}$  and all $\alpha < \alpha_c(\Delta)$.
\end{conj}

The Markov chain analyzed in~\cite{BD97,AL20} is the `down/up' Markov chain: starting from an independent set $I_t \in \cI_k(G)$ at step $t$, pick a uniformly random vertex $v \in I_t$ and a uniformly random vertex $w \in V$.  Let $I' = (I_t \setminus v)\cup w$.  If $I' \in \cI_k(G)$, let $I_{t+1} = I'$; if not, let $I_{t+1} = I_t$.   
It is known that the chain is rapidly mixing for $\alpha< 1/(2\Delta)$, see~\cite{BD97,AL20};  we conjecture that this holds up to the hardness threshold.

\begin{conj}
The down/up Markov chain for sampling from $\cI_{\lfloor \alpha n \rfloor}(G)$ mixes rapidly  for $\alpha < \alpha_c(\Delta)$ and all $G \in \cG_{\le\Delta}$.
\end{conj}

One of the steps of our proof leads to a natural probabilistic conjecture concerning the hard-core model in bounded degree graphs.
\begin{conj}\label{conjLCLT}
Suppose $G$ is a graph on $n$ vertices of maximum degree $\Delta$.  Then if $\lam < \lam_c(\Delta)$ and $k = \lfloor \EE_{G,\lam} | \mathbf I| \rfloor$, we have 
\[ \Pr _{G,\lam} [ |\mathbf I | =k]   = \Omega (n^{-1/2}) \,,  \]
where the implied constant only depends on $\Delta$ and $\lam$ and the expectation and probability are with respect to the hard-core model on $G$ at fugacity $\lam$. 
\end{conj}
Lemma~\ref{lemExistLambda} below gives the weaker bound $\Omega( n^{-1} \log ^{-1} n)$.  A stronger conjecture is that a local central limit theorem for $| \mathbf I|$ holds whenever $\lam < \lam_c(\Delta)$.

Since the extended abstract of this paper appeared at ICALP 2021~\cite{DP21a},  Jain, Perkins, Sah, and Sawhney have resolved Conjectures~\ref{Qdeterministic} and~\ref{conjLCLT} in the affirmative~\cite{jain2021approximate} (in the latter case proving the strong form of the conjecture involving a local central limit theorem).

Next, we conjecture that the explicit value for the computational threshold in Theorem~\ref{thmanti-ferro} is given by the complete graph $K_{\Delta+1}$.  
\begin{conj}\label{conjIsingExtremal}
    Let $\Delta \ge 3$, $B\in(0,B_c( \Delta) )$, and $\lam_c = \lam_c(\Delta, B)$. 
    Then the complete graph $K_{\Delta+1}$ minimizes the occupancy fraction $\alpha_G(B,\lam_c)$ of the anti-ferromagnetic Ising model over the class $\mathcal G_{\Delta}$ of $\Delta$-regular graphs.  That is,
    \[ \alpha_{\mathrm{inf}}(B, \lam_c,\Delta) =  \alpha_{K_{\Delta+1}}(B,\lam_c)  \, .  \]
\end{conj}

The proof of the analogous result for the hard-core model (Theorem~\ref{thmOccMin}) is short and elementary.  For the anti-ferromagnetic Ising model the extremal problem seems much more challenging.

Going beyond the case of $\Delta$-regular graphs, there is a computational threshold for the problem of approximating the partition function of any anti-ferromagnetic $2$-spin system on the class of graphs of maximum degree $\Delta$~\cite{LLY13,SS14,GSV15}. 
Several complications arise in this setting, as the uniqueness region is more subtle to define (using the notion of `up-to-$\Delta$ uniqueness')~\cite{LLY13}, the Ising and hard-core models no longer capture all possible models, and the solution to the extremal problem can change as we have widened the class of graphs to bounded degree instead of $\Delta$-regular.  Nonetheless, the general framework introduced here provides a generic connection between computational thresholds for partition functions of spin systems and for coefficients of partition functions via extremal problems.

\begin{question}
What is the computational threshold for computing coefficients of the partition function of a general $2$-spin anti-ferromagnetic systems on graphs of maximum degree $\Delta$?
\end{question}

Finally, one can ask about the computational complexity of approximating coefficients of partition functions outside of the $2$-spin, antiferromagnetic setting.  In the case of the ferromagnetic Ising model, Carlson, Davies, Kolla, and Perkins~\cite{carlson2021computational} recently proved computational thresholds for approximating coefficients of the partition function; that is, for approximating the partition function of the Ising model at fixed magnetization.  In this case, there can be computational hardness despite the tractability of approximating the Ising model partition function.  The algorithmic results in~\cite{carlson2021computational} build on the techniques here: the parameter regime for which the algorithmic approach succeeds is determined by the solution to an extremal problem of maximizing the magnetization of the Ising model on graphs of maximum degree $\Delta$.  
 Considering anti-ferromagnetic $q$-spin models for $q>2$ would also be interesting, though in this case the complexity of approximating the partition function is not fully resolved.

\section{Algorithms}\label{secAlgorithms}

In this section, we fix $\Delta\ge3$ and $\alpha < \alpha_c(\Delta)$.  We first give an algorithm that, for a graph $G\in\cG_{\le\Delta}$ on $n$ vertices and an integer $k\le\alpha n$, returns an $\eps$-approximate uniform sample from $\cI_k(G)$ and runs in time polynomial in $n$ and $\log(1/\eps)$; this proves the sampling part of Theorem~\ref{thm:main}\ref{itm:algs}. 
We then use this algorithm to approximate $i_k(G)$ using a standard simulated annealing process to prove the approximate counting part of Theorem~\ref{thm:main}\ref{itm:algs}.

Given $\lam\ge0$, let $\bI$ be a random independent set from the hard-core model on $G$ at fugacity $\lam$.
We will write $\Pr_{G,\lam}$ for probabilities over the hard-core measure $\mu_{G,\lam}$, so e.g.\ $\Pr_{G,\lam}(|\bI|=k)$ is the probability that $\bI$ is of size exactly $k$. 
Often we will suppress the dependence on $G$.

A key tool that we use for probabilistic analysis and to approximately sample from $\mu_{G,\lam}$ is the \emph{Glauber dynamics}. 
This is a Markov chain with state space $\cI(G)$ and stationary distribution $\mu_{G,\lam}$. 
Though the algorithm of Weitz~\cite{Wei06} was the first to give an efficient approximate sampling algorithm for $\mu_{G,\lam}$ for $\lam < \lam_c(\Delta)$ and all $G \in \cG_{\le\Delta}$, a randomized algorithm with better running time now follows from recent results showing that the Glauber dynamics mix rapidly for this range of parameters~\cite{ALG20,CLV20,CLV20a}.  
The \emph{mixing time} $\Tmix(\cM,\eps)$ of a Markov chain $\cM$ is the number of steps from the worst-case initial state $I_0$ for the resulting state to have a distribution within total variation distance $\eps$ of the stationary distribution. 
We will use the following result of Chen, Liu, and Vigoda~\cite{CLV20a}, and the sampling algorithm that it implies.

\begin{theorem}[\cite{CLV20a}]\label{thmSampleLambda}
    Given $\Delta\ge 3$ and $\xi\in(0,\lam_c(\Delta))$, there exists $C>0$ such that the following holds. For all $0\le\lam < \lam_c(\Delta)-\xi$ and graphs $G\in\cG_{\le\Delta}$ on $n$ vertices, the mixing time $\Tmix(\cM,\eps)$ of the Glauber dynamics $\cM$ for the hard-core model on $G$ with fugacity $\lam$ is at most $Cn\log(n/\eps) $.  This implies an $\eps$-approximate sampling algorithm for $\mu_{G,\lam}$ for $G \in \cG_{\le\Delta}$ that runs in time $O(n \log n \log (n/\eps))$. 
\end{theorem}

The sampling algorithm follows from the mixing time bound; the extra factor $\log n$ is the cost of implementing one step of Glauber dynamics (which requires reading $O(\log n)$ random bits to sample a vertex uniformly). 
Note that the implicit constant in the running time depends on how close $\lam$ is to $\lam_c(\Delta)$, but in applications of this theorem we will have $\lam \le \lam_c(\Delta)-\xi$ for some fixed $\xi>0$, so that the implicit constant depends only on $\xi$, which in turn depends on $\alpha$.

\subsection{Approximate sampling}

The algorithm Sample-$k$ listed in Algorithm~\ref{algSampleK} uses Theorem~\ref{thmSampleLambda} and a binary search on values of $\lam$ to generate samples from $\cI_k(G)$. 
The algorithm requires access to distributions $\hat\mu_\lam$ which are meant to be approximate versions of the hard-core model on $G$ at fugacity $\lam$. 
For the algorithm to work, we require that there is some value of $\lam$ such that sampling from the hard-core model at fugacity $\lam$ is likely to yield a set of size exactly $k$, and we will establish that this holds for some $\lam$ in a set of size $O(n^2)$ and for a suitable definition of `likely'. 
Quantitatively, these results inform the length $C\log n$ of the for loop and the value of $N$.
An intuitive choice for $\lam$ would be the unique value that makes $k$ the expected size of an independent set from the hard-core model on $G$, and in fact it suffices to find a value sufficiently close to this.

\begin{algorithm}
\SetKwInOut{Input}{input}\SetKwInOut{Output}{output}
\DontPrintSemicolon
\LinesNumbered
\caption{Sample-$k$\label{algSampleK}}
\Input{$\alpha < \alpha_c$; $\eps>0$; $G \in \cG_{\le\Delta}$ of size $n$; integer $k \le \alpha n$}
\Output{$I \in \cI_k(G)$ with distribution within $\eps$ total variation distance of the uniform distribution of $\cI_k(G)$}
\BlankLine
Let $\lam_*=\frac{\alpha}{1-(\Delta+1)\alpha}$\;
For $t = 0, \dots, \lfloor 2\lam_* n^2 \rfloor$, let $\lam_t = t/(2n^2)$\;
Let $\Lambda_0 = \{\lam_t : t = 0, \dots,  \lfloor 2\lam_* n^2 \rfloor \}$ and $N=C'n^2\log\big(\frac{\log n}{\eps}\big)$\;
\For{$i=1, \dots , C\log n$,\label{lineC}}{
Let $\lam$ be a median of the set $\Lambda_{i-1}$\;
Take $N$ independent samples $I_1,\dotsc$, $I_N$ from a distribution $\hat \mu_{\lam}$ on $\cI(G)$\label{lineN}\;
Let $\kappa  = \frac{1}{N}  \sum_{j=1}^N |I_j|$\;
If $| \kappa - k | \le 1/4$ and there exists $j \in \{1, \dots , N\}$ so that $ |I_j| =k$, then output $I_j$ for the smallest such $j$ and \textbf{halt}\;
If $\kappa \le k$, let $\Lambda_i = \{ \lam' \in \Lambda_{i-1} : \lam' > \lam \}$.  If instead $\kappa > k$, let $\Lambda_i = \{ \lam' \in \Lambda_{i-1} : \lam' < \lam \}$\;
}
If no independent set of size $k$ has been obtained by the end of the for loop (or if $\Lambda_j = \emptyset$ at any step), use a greedy algorithm and output an arbitrary $I \in \cI_k(G)$\;
\end{algorithm}

\begin{theorem}\label{thmSampleK}
    Consider the algorithm Sample-$k$. 
    There exist absolute constants $C,C'>0$ such that the following hold for $N=C'n^2\log(\log(n)/\eps)$ as defined in the algorithm.
If the distributions $\hat \mu _{\lam}$ are each within total variation distance $\eps/(2CN\log n)$ of $\mu_{G,\lam}$, then the output distribution of Sample-$k$ is within total variation distance $\eps$ of the uniform distribution of $\cI_k(G)$.
The running time of Sample-$k$ is $O (N\log n \cdot T(n,\eps) )$ where $T(n,\eps)$ is the running time required to produce a sample from $\hat \mu_{\lam}$ satisfying the above guarantee. 
\end{theorem}

The sampling part of Theorem~\ref{thm:main} follows immediately from Theorem~\ref{thmSampleK} since by Theorem~\ref{thmSampleLambda} we can obtain $\eps/(2CN\log n)$-approximate samples from $\mu_{G,\lam}$ in time $O(n \log n \log(n\log n \cdot N/\eps))$. 
Thus, the total running time of Sample-$k$ with this guarantee on $\hat\mu_\lam$ is 
\[ O(N\cdot n\log^2 n\cdot \log(nN/\eps)) \le n^3\log^3n \cdot \mathrm{polylog}\big(\tfrac{\log n}{\eps}\big)\,. \]

Note that the use of Glauber dynamics in Sample-$k$ could be replaced by any other polynomial-time approximate sampler for $\mu_{G,\lam}$, such as the algorithm due to Weitz~\cite{Wei06} based on the method of correlation decay. 
We do however use rapid mixing of the Glauber dynamics in Lemma~\ref{lemExistLambda} below to prove correctness of our algorithm, by establishing a lower bound on the probability that a set sampled from $\mu_{G,\lam}$ has size close to its mean.  A resolution of our Conjecture~\ref{conjLCLT} could imply the necessary probabilistic statement without appealing to the Glauber dynamics, however.

Before we prove Theorem~\ref{thmSampleK}, we collect a number of preliminary results that we will use.  The first is a bound on the probability of getting an independent set of size close to the mean from the hard-core model when $\lam < \lam_c(\Delta)$. We use the notation $n\alpha_G(\lam)$ for the expected size of an independent set from the hard-core model on $G$ at fugacity $\lam$ to avoid ambiguities.

\begin{lemma}\label{lemExistLambda}
For $\Delta\ge 3$ and $\alpha < \alpha_c(\Delta)$, there is a unique $\lam_* < \lam_c(\Delta)$ so that $\alpha_{K_{\Delta+1}}(\lam_*) = \alpha$, and the following holds.
For any $G \in \cG_{\le\Delta}$ on $n$ vertices and any $1 \le k \le \alpha n$, there exists an integer $t \in \{ 0,1, \dots, \lfloor2\lam_*n^2 \rfloor \}$ so that 
\begin{equation}
\label{eqCloseLam}
\left | n\alpha_G(t/(2n^2)) - k  \right | \le 1/2\,.
\end{equation}
Moreover, if  $t$ satisfies~\eqref{eqCloseLam} then
\[ \mu_{G,t/(2n^2)} (\cI_k(G) ) = \Omega \left (  \frac{1}{n \log n } \right) \,.\]
\end{lemma}

To prove this lemma we  need  several more results.  The first is an extremal bound on $\alpha_G(\lam)$ for $G \in \cG_{\le\Delta}$.  The statement of the theorem follows from  a stronger property proved by Cutler and Radcliffe  in~\cite{CR14}; see~\cite{DJPR18a} for discussion.
\begin{theorem}[\cite{CR14}]\label{thmOccMin}
For all $G \in \cG_{\le\Delta}$ and all $\lam \ge 0$,
\[ \alpha_G(\lam) \ge \alpha _{K_{\Delta+1}}(\lam) = \frac{\lam}{1+ \lam (\Delta+1)} \,.\]
\end{theorem}

We next rely on a zero-free region for $Z_G(\lam)$ due to Peters and Regts~\cite{PR19}, so that we can apply the subsequent central limit theorem.
\begin{theorem}[\cite{PR19}]\label{thm:zerofree}
    Let $\Delta\ge 3$ and $\xi\in(0,\lam_c(\Delta))$. Then there exists $\delta>0$ such that for every $G\in\cG_{\le\Delta}$ the polynomial $Z_G$ has no roots in the complex plane that lie within distance $\delta$ of the real interval  $[0, \lam_c(\Delta)-\xi)$.
\end{theorem}

The \emph{probability generating function} of a discrete random variable $\bX$ distributed on the non-negative integers is the polynomial in $z$ given by $f(z) = \sum_{j\ge 0}\Pr(\bX=j)z^j$, and the above result shows that at subcritical fugacity the probability generating function of $|\bI|$ has no zeros close to $1$ in $\mathbb{C}$. 
This allows us to use the following result of Michelen and Sahasrabudhe~\cite{MS19}.

\begin{theorem}[\cite{MS19}]\label{thm:clt}
For $n\ge 1$ let $\bX_n$ be a random variable taking values in $\{0,\dotsc,n\}$ with mean $\mu_n$, standard deviation $\sigma_n$, and probability generating function $f_n$. If $f_n$ has no roots within distance $\delta_n$ of $1$ in $\mathbb{C}$, and $\sigma_n\delta_n/\log n\to\infty$, then $(X_n-\mu_n)/\sigma_n$ converges in distribution to a standard Gaussian.
\end{theorem}

The final tools we need are simple bounds on the variance of the size of an independent set from the hard-core model. A sharper upper bound was subsequently given in~\cite{jain2021approximate}.

\begin{lemma}\label{lem:variance}
    Let $G$ be a graph on $n$ vertices and let $\bI$ be a random independent set drawn from the hard-core model on $G$ at fugacity $\lam$. Then, if  the maximum degree of $G$ is at most $\Delta$ and $M\ge n/(\Delta+1)$ is the size of the largest independent set in $G$, we have
    \[ \frac{\lam}{(1+\lam)^{2+\Delta}}M \le \var(|\bI|) \le n^2\frac{\lam}{1+\lam}\,. \]
\end{lemma}

\begin{proof}
    For the upper bound note that $|\bI|$ is the sum of the indicator random variables $\bX_v$ that the vertex $v\in V(G)$ is in $\bI$. 
    Then because $\Pr(\bX_v=1)\le\lam/(1+\lam)$ for all $v$, from the Cauchy--Schwarz inequality in the form $\cov(\bX_u,\bX_v)^2\le\var(\bX_u)\var(\bX_v)$ we obtain
    \[ \var(|\bI|) = \sum_{u\in V(G)}\sum_{v\in V(G)}\cov(\bX_u,\bX_v) \le n^2\frac{\lam}{1+\lam}\,. \]

    For the lower bound, let $J$ be some fixed independent set in $G$ of maximum size $M$. 
    Now write $\bX=|\bI|$, and let $\bK=\bI\setminus J$. 
    By the law of total variance, 
    \[ \var(\bX) = \EE[\var(\bX|\bK)] + \var(\EE[\bX|\bK]) \ge  \EE[\var(\bX|\bK)]\,. \]
    But we have $\bX= |\bK| + |\bI \cap J|$, and conditioned on $\bK$ the set $\bI \cap J$ is distributed according to the hard-core model on $J\setminus N_G(\bK)$, the subset of $J$ uncovered by $\bK$. 
    Since $J$ is independent, this is a sum of at most $|J|$ independent, identically distributed Bernoulli random variables with probability $\lam/(1+\lam)$. 

    Now, writing $\bU = |J\setminus N_G(\bK)|$ for the number of variables in the sum we have 
    \[ \var(\bX) \ge \EE[\var(\bX|\bK)] = \frac{\lam}{(1+\lam)^2} \EE\bU\,. \] 
    A vertex $u\in J$ is uncovered by $\bK$ precisely when $N(u)\cap \bK=\emptyset$. 
    For each neighbor $v$ of $u$ and for any value of $\bK\setminus\{v\}$, $v$ is occupied with probability at most $\lam/(1+\lam)$, and hence unoccupied with probability at least $1/(1+\lam)$.
    Since $|N(u)|\le\Delta$, the probability that $u$ is uncovered by $\bK$ is at least $(1+\lam)^{-\Delta}$. 
    This means that $\EE\bU \ge |J|(1+\lam)^{-\Delta}$ and hence
    \[ \var(\bX) \ge \frac{\lam}{(1+\lam)^{2+\Delta}}M\,. \]
    The assertion $M\ge n/(\Delta+1)$ follows from the fact that any $n$-vertex graph of maximum degree $\Delta$ contains an independent set of size at least $n/(\Delta+1)$, which is easy to prove by analyzing a greedy algorithm.
\end{proof}

\noindent
Now we are ready to prove Lemma~\ref{lemExistLambda}.

\begin{proof}[Proof of Lemma~\ref{lemExistLambda}]
A standard calculation gives
    \[ \frac{\partial}{\partial\lam}\alpha_G(\lam) = \frac{1}{n}\frac{\partial}{\partial\lam}\frac{\lam Z_G'(\lam)}{Z_G(\lam)} = \frac{1}{n\lam}\var(|\bI|)\,, \]
  and  so Lemma~\ref{lem:variance} gives that $0 < \alpha'_G(\lam) \le n$ for all $\lam>0$.

 Next, let $\lam_* < \lam_c(\Delta)$ be the solution to the equation $\alpha_{K_{\Delta+1}}(\lam_*) = \alpha$.  
 This means 
 \[ \lam_* = \frac{\alpha}{1-(\Delta+1)\alpha}\,, \]
 as defined in Sample-$k$.
 The fact that $\lam_* < \lam_c(\Delta)$ follows from the fact that $\alpha < \alpha_c(\Delta) = \alpha_{K_{\Delta+1}}(\lam_c(\Delta))$, and that occupancy fractions are strictly increasing. Then using Theorem~\ref{thmOccMin} we have that 
  \begin{equation}
     \alpha_G(\lam_*) \ge \alpha_{K_{\Delta+1}}(\lam_*) =\alpha \, ,
     \end{equation}
  and so there exists $\lam \in (0, \lam_*]$ such that $n \alpha_G(\lam) = k$.  
  Using the upper bound on $ \alpha'_G(\lam)$, we see that  as $\lam$ increases over an interval of length $1/(2n^2)$, the function $n\alpha_G(\lam)$ can increase by at most $1/2$.
 Hence, there is at least one integer $t \in \{ 1, \dots , \lfloor 2\lam_* n^2 \rfloor \}$ such that $| n \alpha_G(t/(2n^2)) - k | \le 1/2$.

    The second statement of Lemma~\ref{lemExistLambda} follows from a central limit theorem for $|\bI|$ and rapid mixing of the Glauber dynamics. 
    There is a close connection between zeros of the probability generating function of $|\bI|$ and the zeros of the partition function itself. 
    The probability generating function of $|\bI|$ is
    \[ f(z) = \sum_{j\ge 0}\Pr_\lam(|\bI|=j)z^j = \sum_{j\ge 0}\frac{i_j(G)\lam^j z^j}{Z_G(\lam)} = \frac{Z_G(\lam z)}{Z_G(\lam)}\,. \]
    Then for $\lam$ such that $Z_G(\lam)\ne 0$, $z$ is a root of $f$ if and only if $z\lam$ is a root of $Z_G(\lam)$.
    By our assumptions on $t$, when $\lam=t/(2n^2)$ Theorem~\ref{thm:zerofree} gives the existence of $\delta>0$ such that for all $G\in\cG_{\le\Delta}$ there are no complex zeros of $f$ within distance $\delta/\lam$ of $1$. 
    This is because Theorem~\ref{thm:zerofree} means that $Z_G(z\lam)=0$ implies $|z\lam - \lam| \ge \delta$.
    The condition of Theorem~\ref{thm:clt} which states that $\sigma_n\delta_n /\log n\to\infty$ is met because $\lam < \lam_c(\Delta) \le 4$ and so
    \[ \sigma_n\delta_n \ge \sqrt{\frac{\lam}{(1+\lam)^{2+\Delta}}\frac{n}{\Delta+1}}\cdot\frac{\delta}{\lam} \ge \Omega\Big(\sqrt{n/\lam}\Big)>\omega(\log n)\,. \]

    Now, let $\lam=t/(2n^2)$ and suppose that~\eqref{eqCloseLam} holds, meaning that $k$ is within $1/2$ of $n \alpha_G(\lam)$. 
    The standard deviation of the size of a set drawn from $\mu_{G,\lam}$ is at least a constant, which follows from the lower bound in Lemma~\ref{lem:variance} of $\Omega(\sqrt{\lam n})$ and the fact that $\lam\ge \Omega(1/n)$. 
    To see this, note that $\alpha_G(\lam) \le \lam/(1+\lam)$ for any graph and non-negative fugacity $\lam$. This holds because when $\bI$ is a random independent set from the hard-core model, conditioned on a vertex $v$ having no neighbors in $\bI$, $v\in\bI$ with probability $\lam/(1+\lam)$. If $v$ has a neighbor in $\bI$ then $v\notin\bI$ with probability $1$, and the bound follows. 
    Then using~\eqref{eqCloseLam}, we have
    \[ n\frac{\lam}{1+\lam} \ge n\alpha_G(\lam) \ge k-1/2 \ge 1/2\,, \]
    and so $\lam \ge \Omega(1/n)$.
    We deduce that $k$ is within some constant number $r>0$ of standard deviations of the mean size $n\alpha_G(\lam)$.
    The central limit theorem and standard properties of the normal distribution mean that there are constants $\rho>0$ (small enough as a function of $r$) and $n_0$ such that for all $n\ge n_0$, with probability at least $\rho$, $|\bI|$ is at least $r$ standard deviations below the mean, and similarly with probability at least $\rho$ it is at least $r$ standard deviations above the mean. 
    So we have $\Pr_{G,\lam}(|\bI|\ge k)\ge\rho$ and $\Pr_{G,\lam}(|\bI|\le k)\ge\rho$.
    
    The transition probabilities when we are at state $I$ in the Glauber dynamics are given by the following random experiment.
    Choose a vertex $v\in V(G)$ uniformly at random and let
    \[ I' = \begin{cases}
        I\cup\{v\} & \text{with probability $\lam/(1+\lam)$}\,, \\
        I\setminus\{v\} & \text{with probability $1/(1+\lam)$}\,.
    \end{cases} \]
    Now if $I'$ is independent in $G$ move to state $I'$, otherwise stay in state $I$.
    This means that the sequence of sizes of set visited must take consecutive integer values.
    By Theorem~\ref{thmSampleLambda}, there is a constant $C''$ such that from an arbitrary starting state, in $C''n\log n$ steps the distribution $\pi$ of the current state is within total variation distance $\rho/2$ of the hard-core model. Then the following statements hold.
    \begin{enumerate}[label=\textup{(\roman*)}]
        \item Starting from an independent set of size at most $k$, with probability at least $\rho/2$ the state after $C''n\log n$ steps is an independent set of size at least $k$.
        \item Starting from an independent set of size at least $k$, with probability at least $\rho/2$ the state after $C''n\log n$ steps is an independent set of size at most $k$.
    \end{enumerate}
    Consider starting from an initial state distributed according to $\mu_{G,\lam}$. Then every subsequent state is also distributed according to $\mu_{G,\lam}$, and the above facts mean that for any sequence of $C''n\log n$ consecutive steps, with probability at least $\rho/2$ we see a state of size exactly $k$. 
    Recalling that $\lam=t/(2n^2)$, this immediately implies that
    \[ \mu_{G,t/(2n^2)}(\cI_k(G)) \ge \frac{\rho}{2C''n\log n}\,, \]
    as required. 
\end{proof}

\begin{proof}[Proof of  Theorem~\ref{thmSampleK}]
We first prove the theorem under the assumption that each $\hat\mu_\lam$ is exactly the hard-core measure $\mu_{G,\lam}$, taking note of how many times we sample from any $\hat\mu_\lam$.

We say a \emph{failure} occurs at step $i$ in the FOR loop if either of the following occur:
\begin{enumerate}[label=\textup{(\alph*)}]
\item  $| n\alpha_G(\lam) - \kappa| > 1/4$.
\item  $| n\alpha_G(\lam) - k| \le 1/2$ but the algorithm did not output an independent set of size $k$ in step $i$.
\end{enumerate}
We show that the probability that a failure occurs at any time during the algorithm is at most $\eps/2$.  By a union bound, it is enough to show that the probability of either type of failure at a given step $i$ is at most $\frac{\eps}{4C\log n}$.  

Consider an arbitrary step $i$ with its value of $\lam$. To bound the quantity $\Pr(|n\alpha_G(\lam) - \kappa| > 1/4)$, note that $\kappa$ is the mean of $N$ independent samples from $\hat\mu_\lam$, which we currently assume to be $\mu_{G,\lam}$. Then we have $\EE\kappa = n\alpha_G(\lam)$ and Hoeffding's inequality gives
$ \Pr(|n\alpha_G(\lam)-\kappa| > 1/4) \le 2e^{-N/(8n^2)}$, 
so for this to be at most $\eps/(4C\log n)$ we need only
$N\ge \Omega\big(n^2\log\big(\tfrac{\log n}{\eps}\big)\big)$.

To bound the probability that the current step involves $\lam$ such that $|n\alpha_G(\lam) - k| \le 1/2$, but we fail to get a set of size $k$ in the $N$ samples, observe that we have $N$ independent trials for getting a set of size $k$, and each trial succeeds with probability $p\ge c/(n\log n)$ by Lemma~\ref{lemExistLambda}. 
Then the probability we see no successful trials is $\big(1-\frac{c}{n\log n}\big)^N$,
which is at most $\eps/(4C\log n)$ for 
$N \ge \Omega\big(n\log n \cdot \log\big(\tfrac{\log n}{\eps}\big)\big)$.  
Thus, we can take $N = \Theta\big(n^2\log\big(\tfrac{\log n}{\eps}\big)\big)$, as in line~\ref{lineN} of Sample-$k$.

Next we show that in the event that no failure occurs during the running of the algorithm, the algorithm outputs an independent set $I$ with distribution within $\eps/2$ total variation distance of the uniform distribution on $\cI_k(G)$. 

We first observe that if no failure occurs, the algorithm at some point reaches a value of $\lam$ so that  $|n\alpha_G(\lam) - k| \le 1/2$. 
This is a simple consequence of Lemma~\ref{lemExistLambda}, which means there exists some $t$ with this property, and the binary search structure of the algorithm. 
In particular, in each iteration of the FOR loop, at line (e) the size of the set $\Lambda_i$ being searched goes down by (at least) half. 
Conditioned on no failures, the search also proceeds in the correct half of $\lam_i$ because 
we search the upper half only when $\kappa < k - 1/4$ and so conditioned on no failure we have 
$n\alpha_G(\lam) \le \kappa + 1/4 < k$ and hence using a larger value of $\lam$ must bring $n\alpha_G(\lam)$ closer to $k$. The case $\kappa > k + 1/4$ is similar.
This means that, conditioned on no failures, the algorithm must reach a value of $\lam$ such that $|n\alpha_G(\lam)-k|\le 1/4$.

Note that $\mu_{G,\lam}$ conditioned on getting a set of size exactly $k$ is precisely the uniform distribution on $\cI_k(G)$, hence if the algorithm outputs an independent set of size $k$ during the FOR loop, its distribution is exactly uniform distribution on $\cI_k(G)$.
Thus, under the assumption that each $\hat\mu_\lam$ is precisely $\mu_{G,\lam}$ we have shown that with probability at least $1-\eps/2$ no failures occur, and hence a perfectly uniform sample from $\cI_k(G)$ is output during the FOR loop. 

We do not need access to an efficient exact sampler for $\mu_{G,\lam}$, and can make do with the approximate sampler from Theorem~\ref{thmSampleLambda}. 
One interpretation of total variation distance is that when each $\hat\mu_\lam$ has total variation distance at most $\xi$ from $\mu_{G,\lam}$, there is a coupling between $\hat\mu_\lam$ and $\mu_{G,\lam}$ such that the probability they disagree is at most $\xi$. 
Then to prove Theorem~\ref{thmSampleK} we consider a third failure condition: that during any of the calls to a sampling algorithm for any $\hat\mu_\lam$ the output differs from what would have been given by $\mu_{G,\lam}$ under this coupling. 
Since we make at most $CN\log n$ calls to such sampling algorithms, provided $\xi \le \eps/(2CN\log n)$ the probability of any failure of this kind is at most $\eps/2$. 
Together with the above proof for samples distributed exactly according to $\mu_{G,\lam}$ which successfully returns uniform samples from $\cI_k(G)$ with probability $1-\eps/2$, we have now shown the existence of a sampler that with probability $1-\eps$ returns uniform samples from $\cI_k(G)$, and makes at most $CN\log n$ calls to a $\eps/(2CN\log n)$-approximate sampler for $\mu_{G,\lam}$ (at various values of $\lam$). 
Interpreting this in terms of total variation distance, this means we have an $\eps$-approximate sampler for the uniform distribution on $\cI_k(G)$ with running time $O(N\log n \cdot T(n,\eps))$.
\end{proof}

\subsection{Approximate counting}\label{sec:annealing}

Given a graph $G=(V,E)$ on $n$ vertices and $j\ge 0$, let $f_j(G) = (j+1)i_{j+1}(G)/i_j(G)$. This $f_j(G)$ has an interpretation as the \emph{expected free volume}  over a uniform random independent set $\bJ \in \cI_j(G)$, that is, $f_j = \EE|V\setminus(\bJ\cup N(\bJ))|$. This holds because each vertex in $V\setminus(\bJ\cup N(\bJ))$ can be added to $\bJ$ to make an independent set of size $j+1$, and each such set is counted $j+1$ times in this way. 
Then by a simple telescoping product we have
\begin{equation}\label{eqAnnealing}
     i_k(G) = \prod_{j=0}^{k-1}\frac{f_j(G)}{j+1}\,,
\end{equation}
and hence if for $0\le j\le k-1$ we can obtain a relative $\eps/k$ approximation to $f_j$ in time polynomial in $n$ and $1/\eps$ then we can obtain a relative $\eps$-approximation to $i_k(G)$ in time polynomial in $n$ and $1/\eps$.
By the definition of $f_j$ as an expectation over a uniform random independent set of size $j$, we can use an efficient sampling scheme for this distribution to approximate $f_j$, which is provided by Theorem~\ref{thmSampleK}. 
That is, by repeatedly sampling independent sets of size $j$ approximately uniformly and recording the free volume we can approximate the expected free volume $f_j(G)$, and hence the corresponding term of the product in~\eqref{eqAnnealing}. 
Doing this for all $0\le j\le k-1$ thus provides an approximation to $i_k(G)$. 
This scheme is an example of \emph{simulated annealing}, which can be used as a general technique for obtaining approximation algorithms from approximate sampling algorithms. For more details, see e.g.~\cite{JS97,SJ89}. 
Here the integer $j$ is playing the role of inverse temperature, and we approximate $i_k(G)$ by estimating $f_j(G)$ (by sampling from $\cI_j(G)$) with the cooling schedule $j=0,1,\dotsc,k-1$. 
We expect that a more sophisticated cooling schedule can be used to decrease the running time of our reduction, see for example~\cite{HK20}.

Since this annealing process is standard, we sketch a simple version of the method. Suppose that for all $0\le j \le k-1$ we have a randomized algorithm that with probability at least $1-\delta'$ returns a relative $\eps/k$-approximation $\hat t_j$ to $f_j(G)/(j+1)$ in time $T'$. 
Then~\eqref{eqAnnealing} implies that with probability at least $1-k\delta'$, the product $\hat\imath_k = \prod_{j=0}^{k-1}\hat t_j$ is a relative $\eps$-approximation to $i_k(G)$, and this takes time $kT'$ to compute. 
For the FPRAS in Theorem~\ref{thm:main}, it therefore suffices to design the hypothetical algorithm with $\delta'=1/(4k)$ and $T'$ polynomial in $n$ and $1/\eps$. 

First, suppose that we have access to an exactly uniform sampler for $\cI_j(G)$ for $0\le j\le k-1$, but impose the smaller failure probability bound of $\delta'/2$. 
Then, for each $j$, let $\hat t_j$ be the sample mean of $m$ computations of $|V\setminus(\bJ\cup N(\bJ))|/(j+1)$ where $\bJ$ is a uniform random independent set of size $j$.
We note that as a random variable $|V\setminus(\bJ\cup N(\bJ))|/(j+1)$ has a range of at most $j\Delta/(j+1)$ in a graph of maximum degree $\Delta$ because $0\le |N(\bJ)| \le j\Delta$, and for $j\le k-1$ and 
\[ k\le\alpha n < \alpha_c(\Delta) n < \frac{e}{1+e} \frac{n}{\Delta}\,,\] 
we have 
\[ \frac{|V\setminus(\bJ\cup N(\bJ))|}{j+1} \ge \frac{n-j(\Delta+1)}{j+1} \ge \frac{\Delta}{e}-1\,.\]
Let $S_j$ be the mean of $m$ samples of $|V\setminus(\bJ\cup N(\bJ))|/(j+1)$.
Then, using that for $\eps'\le 1$ it suffices to ensure $|S_j - \mu|\le \eps'\mu/2$ for $S_j$ to be a relative $\eps'$-approximation to $\mu$, by Hoeffding's inequality, 
\[ m \ge \Omega(\eps^{-2}k^2\log(1/\delta')) = \Omega(\eps^{-2}k^2\log k) \]
samples are sufficient to obtain the required approximation accuracy $\eps'$ with the required success probability $1-\delta'/2$. 
Since we do not have an exact sampler, we use the approximate sampler obtained in this section with total variation distance $\delta'/2$. 
Using the coupling between the exact and the approximate sampler that we used in the proof of Theorem~\ref{thmSampleK}, this suffices to obtain the required sampling accuracy with failure probability at most $\delta'$. 
Recalling that $k\le n$, it is now simple to check that the running time of the entire annealing scheme is polynomial in $n$ and $1/\eps$.
This completes the proof of Theorem~\ref{thm:main}\ref{itm:algs}.

\section{Hardness}\label{secHardness}

Let $\mathrm{IS}(\alpha,\Delta)$ be the problem of computing $i_{\lfloor \alpha n \rfloor}(G)$ for a $\Delta$-regular graph $G$ on $n$ vertices, and let $\mathrm{HC}(\lam,\Delta)$ be the problem of computing $Z_G(\lam)$ for a $\Delta$-regular graph $G$.  

To prove hardness we will use a notion similar to an `approximation-preserving reduction' from~\cite{dyer2004relative}.  
For a fixed error parameter $\eps>0$, we reduce $\mathrm{HC}(\lam,\Delta)$ to $\mathrm{IS}(\alpha,\Delta)$ by giving an algorithm that takes an instance $G$ of $\mathrm{HC}(\lam,\Delta)$ and computes a relative $\eps$-approximation to $Z_G(\lam)$ in time polynomial in $|G|$ by constructing some instance $G'$ of $\mathrm{IS}(\alpha,\Delta)$ and asking an oracle for an $\eps/2$-approximation to $i_{\lfloor \alpha|G'|\rfloor}(G')$. 
We ensure that $|G'|$ is polynomial in $|G|$, and if we could additionally bound $|G'|$ by a polynomial in $1/\eps$ we would have an approximation-preserving reduction as in~\cite{dyer2004relative}. 
Note that this gives a stronger result than the stated hardness in Theorem~\ref{thm:main}, as our reduction gives hardness for any fixed constant error $\eps$.

The main result of this section is the following.
\begin{theorem}\label{thmAPreduction}
For every $\eps>0$, $\Delta \ge 3$ and $\alpha \in (\alpha_c(\Delta), 1/2)$, there exists $\lam > \lam_c(\Delta)$ and a polynomial-time algorithm with the following properties. 
\begin{enumerate}[label=\textup{(\roman*)}]
		\item Given an instance $G$ of $\mathrm{HC}(\lam,\Delta)$, the algorithm constructs an instance $G'$ of $\mathrm{IS}(\alpha,\Delta)$ with size polynomial in $|G|$.
			\item Given a relative $\eps/2$-approximation of $i_{\lfloor\alpha |G'|\rfloor}(G')$, a relative $\eps$-approximation of $Z_G(\lam)$ can be computed in polynomial time.
\end{enumerate}
\end{theorem}

\noindent
Theorem~\ref{thmAPreduction} immediately implies the hardness part of Theorem~\ref{thm:main}, as we recall the hardness of obtaining a constant factor approximation for the problem $\mathrm{HC}(\lam,\Delta)$. 
\begin{theorem}[\cite{Sly10,galanis2016inapproximability,SS14}]\label{thmHChard}
The following holds for any fixed $\eps >0$ and $\lam>\lam_c(\Delta)$. 
Unless NP=RP, there is no polynomial-time algorithm that outputs a relative $\eps$-approximation of $Z_G(\lam)$ on the class of $\Delta$-regular graphs.
\end{theorem}

\begin{proof}[Proof of Theorem~\ref{thmAPreduction}]

Fix $\Delta \ge 3$, and let $\alpha \in (\alpha_c(\Delta) ,1/2 )$ be given.  We will construct a $\Delta$-regular graph $H$ on $n_H$ vertices such that for some value $\lam \in (\lam_c(\Delta),\infty)$ we have $\alpha_H(\lam) = \alpha$.
Our reduction is then as follows: given a $\Delta$-regular graph $G$ on $n$ vertices and $\eps>0$, let $G'$ be the disjoint union of $G$ with $rH$, the graph of $r$ disjoint copies of $H$, with $r= \lceil C\Delta n^2/\eps \rceil$ for some absolute constant $C$.  Let $N= | V(G')| = n+ r n_H$.  We will prove that
\begin{equation}
\label{eqReductionEquation}
e^{-\eps/2}  \frac{ i_{k } (G')}{  i_k (rH)   }   \le  Z_G(\lam) \le e^{\eps/2}  \frac{ i_{k } (G')}{   i_k (rH)   }   \,,
\end{equation}
where $k = \lfloor \alpha N \rfloor$.

Since $G'$ can be constructed and  $i_k (rH)$  computed in time polynomial in $n$, this provides the desired reduction.  
What remains is to construct the graph $H$ satisfying $\alpha_H(\lam) = \alpha$ and then to prove~\eqref{eqReductionEquation}. 

\subsection*{Constructing \texorpdfstring{$H$}{H}}

The graph $H = H_{a,b}$ will consist of the union of $a$ copies of the complete bipartite graph $K_{\Delta,\Delta}$ and $b$ copies of the complete graph $K_{\Delta+1}$.  Clearly $H$ is $\Delta$-regular.   
Since the occupancy fraction of any graph is a strictly increasing function of $\lam$, and the relevant occupancy fractions satisfy $\alpha_{K_{\Delta+1}}(\lam_c(\Delta)) = \alpha_c(\Delta)$ and $\lim_{\lam \to \infty} \alpha_{K_{\Delta,\Delta}}(\lam) =1/2$, we see that there exist integers $a,b \ge 0$ (with at least one positive) and $\lam > \lam_c(\Delta)$ so that $\alpha_{H_{a,b}}(\lam) =  \alpha$. 
A given pair $(a,b)$ provides a suitable $H_{a,b}$ when 
\[ \alpha_{H_{a,b}}(\lam_c(\Delta)) < \alpha < \lim_{\lam\to\infty}\alpha_{H_{a,b}}(\lam) = \frac{a\Delta +b}{2a\Delta + b(\Delta+1)}\,,\]
and hence it can be shown that for all $\Delta\ge 3$ one of the pairs $(0,1)$, $(1,16)$, $(1,6)$, $(1,3)$, $(2,3)$, $(2,1)$, $(1,0)$ suffices for $(a,b)$, and a suitable pair is easy to find efficiently.
This provides us with the desired graph $H$.  From here on, fix these values $a,b, \lam$ and let $n_H = 2 a \Delta + b (\Delta+1)$. 

\subsection*{Proving~\texorpdfstring{\eqref{eqReductionEquation}}{the desired approximation ratio}}

We now form $G'$ by taking the union of $G$ (a $\Delta$-regular graph on $n$ vertices) and $r$ copies of $H$.  Let $N = n + r n_H$ be the number of vertices of $G'$, and write $k = \lfloor \alpha N \rfloor$.  Let $rH$ be the graph consisting of the disjoint union of $r$ copies of $H$.  We can write:
\begin{align*}
i_k(G') &= \sum_{j=0}^n i_j(G)  i_{k-j} (rH) = i_k(rH)  \sum_{j=0}^n i_j(G)  \frac{i_{k-j} (rH) }{ i_k(rH)} \,.
\end{align*}
Now to prove~\eqref{eqReductionEquation} it suffices to show that for $r \ge C\Delta n^2 /\eps$ and $0 \le  j \le n$, we have 
\begin{equation}
e^{-\eps/2} \lam^j  \le  \frac{i_{k-j} (rH) }{ i_k(rH)}  \le e^{\eps/2} \lam^j  \,.
\end{equation}
That is, we are done if the ratios $i_{k-j}(rH)/i_k(rH)$ are well-approximated by $\lam^j$. 
We should expect this to hold because of a  local central limit theorem (CLT) for the hard-core model $\mu_{rH,\lambda}$ on $rH$, and the fact that $r\gg n$. 
Since $rH$ is a disjoint union of $r$ copies of $H$ we can use a standard local limit theorem for sums of i.i.d.\ random variables, see the upcoming Theorem~\ref{thm:gnedenko}. 
A local CLT for $\mu_{rH,\lam}$ is related to the ratios $i_{k-j}(rH)/i_k(rH)$ by a small calculation.
We have the exact formula (for any $0\le j\le k$)
\begin{align*}
i_{k-j} (rH) =  \frac{Z_{rH}(\lam)}{\lam^{k-j} } \Pr_{rH,\lam} ( |\mathbf I | = k-j)
\end{align*}
and so
\begin{align*}
  \frac{i_{k-j} (rH) }{ i_k(rH)} &= \lam^j \frac{  \Pr_{rH,\lam} ( |\mathbf I | = k-j)    }{ \Pr_{rH,\lam} ( |\mathbf I | = k)      } \,,
\end{align*}
where $\Pr_{rH,\lam}$ denotes probabilities with respect to an independent set $\mathbf I$ drawn according to the hard-core model on $rH$ at fugacity $\lam$. 
It is then enough to show 
\begin{align*}
e^{-\eps/2}  \le  \frac{  \Pr_{rH,\lam} ( |\mathbf I | = k-j)    }{ \Pr_{rH,\lam} ( |\mathbf I | = k)     }  \le e^{\eps/2} \,.
\end{align*}
This will follow from a local central limit theorem (e.g.~\cite{Gne48}) since $|\mathbf I|$ is the sum of $r$ i.i.d.\ random variables and the fact that $\EE_{rH,\lam}|\bI|$ is close to both $k$ and $k-j$. 
The following theorem gives us what we need.

\begin{theorem}[Gnedenko~\cite{Gne48}]\label{thm:gnedenko}
    Let $X_1, \dotsc, X_r$ be i.i.d.\ integer valued random variables with mean $\mu$ and variance $\sigma^2$, and suppose that the support of $X_1$ includes two consecutive integers. Let $S_r = X_1 + \dotsb + X_r$.  Then
    \begin{align*}
    \Pr(S_r =k) &=  \frac{1}{\sqrt{ 2 \pi r} \sigma} \exp \left[ -(k-r\mu)^2/(2r \sigma^2)   \right ] + o(r^{-1/2})\,,
    \end{align*}
    with the error term $o(r^{-1/2})$ uniform in $k$.
    \end{theorem}

This immediately implies that with $\mu$ and $\sigma^2$ the mean and standard deviation of the hard-core model on $H$ at fugacity $\lam$ and terms $o(1)\to 0$ as $r\to \infty$,
\begin{align*}
    \frac{  \Pr_{rH,\lam} ( |\mathbf I | = k-j)    }{ \Pr_{rH,\lam} ( |\mathbf I | = k) } = \frac{e^{-[j^2-2(k-r\mu)j]/(2r\sigma^2)}+o(1) \cdot e^{(k-r\mu)^2/(2r\sigma^2)}}{1+o(1)\cdot e^{(k-r\mu)^2/(2r\sigma^2)}}\,.
\end{align*}
It therefore suffices to show that for large enough $r$, namely $r\ge C\Delta n^2/\eps$, we can make $[j^2-2(k-r\mu)j]/(2r\sigma^2)$ smaller than $\eps/4$, show that $(k-r\mu)^2/(2r\sigma^2)$ is bounded above by some absolute constant, and take $r$ large enough in terms of the constant $\eps$ such that the entire right-hand side lies in $[e^{-\eps/2}, e^{\eps/2}]$.
Note that $\mu=\alpha n_H$, and by Lemma~\ref{lem:variance} we have for all $\Delta\ge 3$ (and any choices of $\alpha,\lam,a,b$ made according to our conditions),
\[ \sigma^2 \ge\frac{\lam}{(1+\lam)^{2+\Delta}}(a\Delta+b) \ge \frac{\lam_c(\Delta)}{(1+\lam_c(\Delta))^{2+\Delta}}(a\Delta+b)\ge \frac{0.00384}{\Delta}\,. \]
Since $k=\lfloor\alpha N\rfloor = \lfloor \alpha n + r\alpha n_H\rfloor$, we then have $(k-r\mu)^2 \le \alpha^2 n^2<n^2$, and hence $\frac{(k-r\mu)^2}{2r\sigma^2} \le C'\Delta \frac{n^2}{r}$, 
where $C'$ is an absolute constant. 
Now since $0\le j\le n$ we also have
\[ \left|\frac{j^2-2(k-r\mu)j}{2r\sigma^2}\right| \le C'\Delta \frac{n^2}{r}\,. \]
This means that provided we take $C$ to be a large enough absolute constant and $r \ge C \Delta n^2/\eps$, we have~\eqref{eqReductionEquation} as required.
\end{proof}

We note that the known inapproximability for $\mathrm{HC}(\lam,\Delta)$ is substantially stronger than the constant factor that we state in Theorem~\ref{thmHChard} (see~\cite{GSV15,SS14}), and in fact for fixed $\Delta$ and $\lam > \lam_c(\Delta)$ there exists a small enough constant $c$ such that there is no polynomial-time algorithm achieving a $e^{c n}$-relative approximation to $Z_G(\lam)$ on $n$-vertex graphs $G$ of maximum degree $\Delta$, unless NP=RP\@. It would be interesting to prove similarly strong hardness for $\mathrm{IS}(\alpha,\Delta)$ when $\alpha > \alpha_c(\Delta)$, but our approach does not seem to extend straightforwardly to such a result.
When we apply Theorem~\ref{thm:gnedenko} we do not have an explicit rate of convergence for the error term $o(r^{-1/2})$, and in the case that the desired error $\eps$ depends on $n$ (and hence $r$) this must be overcome. Moreover, the fact that $r \ge \Omega(n^2)$ in the proof of Theorem~\ref{thmAPreduction} is an obstacle in the case $\eps$ is exponential in $n$.

\section{Triangle-free graphs}\label{secTriangles}

In this section we briefly describe the modifications to the proofs in Sections~\ref{secAlgorithms} and~\ref{secHardness} to yield Theorem~\ref{thmTriangleFree}, an analogue of Theorem~\ref{thm:main} for triangle-free graphs.

\subsection{Algorithms}
We use the following lower bound on the occupancy fraction of triangle-free graphs, given in~\cite[Theorem 3]{DJPR17a}.

\begin{theorem}[\cite{DJPR17a}]\label{thmOccMinTriangleFree0}
For every graph $G$ of maximum degree $\Delta$ and every $\lam>0$,
\[ \alpha_G(\lam) \ge \frac{\lam}{1+\lam} \frac{ W (\Delta \log(1+\lam))  }{ \Delta \log(1+\lam)  } \,,\]
where $W (\cdot)$ is the Lambert $W$-function, satisfying $W(z) e^{W(z)}= z$. 
\end{theorem}

\noindent
From Theorem~\ref{thmOccMinTriangleFree0} we can obtain the following.
\begin{cor}[\cite{DJPR17a}]\label{thmOccMinTriangleFree}
For every $\delta>0$, there is $\Delta_0$ large enough so that for every $\Delta\ge \Delta_0$, and every triangle-free $G \in \cG_{\le\Delta}$,
\[ \alpha_G(\lam_c(\Delta)-1/\Delta^2) \ge \frac{1 - \delta}{\Delta} \,.\]
\end{cor}
\begin{proof}
Recall that $\lam_c(\Delta) = \frac{(\Delta-1)^{\Delta-1}}{(\Delta-2)^\Delta}$. Using Taylor expansion we have $ \lam_c(\Delta) = \frac{e}{\Delta} + O(\Delta^{-2})$.
From Theorem~\ref{thmOccMinTriangleFree0} we have (with asymptotics as $\Delta \to \infty$), 
\begin{align*}
\alpha_G(\lam_c(\Delta)-1/\Delta^2) &\ge \frac{\lam_c(\Delta)-1/\Delta^2}{1+\lam_c(\Delta)-1/\Delta^2} \frac{ W (\Delta \log(1+\lam_c(\Delta)-1/\Delta^2))  }{ \Delta \log(1+\lam_c(\Delta)-1/\Delta^2)  }  \\
&= (1+o(1)) \frac{e}{\Delta}   \frac{  W \left(  \Delta \left(  \frac{e}{\Delta} +O(\Delta^{-2}) \right)  \right)  }{ \Delta \left(  \frac{e}{\Delta} +O(\Delta^{-2}) \right)  } \\
&= (1+o(1)) \frac{1}{\Delta}   W(e+ o(1)) \\
&=  (1+o(1)) \frac{1}{\Delta} \,,
\end{align*}
and so for $\Delta $ large enough as a function of $\del$, we have $\alpha_G(\lam_c(\Delta)-1/\Delta^2) \ge \frac{1-\del}{\Delta}$.
\end{proof}

Now the algorithm for Theorem~\ref{thmTriangleFree} is essentially the same as for Theorem~\ref{thm:main}, but since we assume the graph $G$ is triangle free we can use a stronger lower bound on the occupancy fraction than Theorem~\ref{thmOccMin}. 
Let $\delta>0$ and  $\alpha <(1-\delta)/\Delta$ as in Theorem~\ref{thmTriangleFree}. Then Corollary~\ref{thmOccMinTriangleFree} means that for $\Delta\ge\Delta_0$ and any triangle-free graph $G\in\cG_{\le\Delta}$ we have 
\[ \alpha_G(\lam_c(\Delta)-1/\Delta^2) \ge \frac{1-\delta}{\Delta} > \alpha\,. \]
But occupancy fractions are continuous and strictly increasing, so with $\lam_* = \lam_c(\Delta) - 1/\Delta^2$ there exists $\lam\in(0,\lam_*]$ such that $k=n\alpha_G(\lam)$, as in the proof of Lemma~\ref{lemExistLambda} but permitting larger $\alpha$.
The analysis of the algorithm can then proceed exactly as in the proofs of Lemma~\ref{lemExistLambda} and Theorem~\ref{thm:main}, completing the proof of Theorem~\ref{thmTriangleFree}\ref{itmTriangleFreeAlgs}.

\subsection{Hardness}
The proof of hardness for triangle-free graphs is the same as that for general graphs, but we replace $K_{\Delta+1}$ with a (constant-sized) random regular graph in the construction. 
Bhatnagar, Sly, and Tetali~\cite{BST16} showed that the local distribution of the hard-core model on the random regular graph converges to that of the unique translation-invariant hard-core measure on the infinite regular tree for a range of $\lam$ including $\lam=\lam_c(\Delta)$. 
This means that if $K$ is a random $\Delta$-regular graph on $n$ vertices and $\alpha_{T_\Delta}$ denotes the occupancy fraction of the unique translation-invariant hard-core measure on the infinite $\Delta$-regular tree (see~\cite{BST16,DJPR17a}) we have with probability $1-o_n(1)$,
\[ \alpha_K(\lam_c(\Delta)) = \alpha_{T_\Delta}(\lam_c(\Delta)) + o_n(1) = \frac{1+o_{n,\Delta}(1)}{\Delta}\,, \]
where $o_n(1)\to 0$ as $n\to \infty$ and $o_{n,\Delta}(1)\to 0 $ as both $n$ and $\Delta$ tend to infinity.
Thus, for fixed $\delta\in(0,1)$, there is $n_0 = n_0(\delta)$ and $\Delta_0=\Delta_0(\delta)$ such that with probability at least $1-\delta$ a random $\Delta$-regular graph $K$ on $n_0$ vertices has $\alpha_K(\lam_c(\Delta)) \le (1+\delta)/\Delta$.
This means that in time bounded by a function of $\delta$ an exhaustive search over $\Delta$-regular graphs on $n_0$ vertices must yield a graph $K$ with the property $\alpha_K(\lam_c(\Delta)) \le (1+\delta)/\Delta$.
Now we replace $K_{\Delta+1}$ with the random $\Delta$-regular graph $K$ in the proof above, which for $\Delta\ge\Delta_0$ allows us to work with any $\alpha \in ((1+\delta)/\Delta,1/2)$ by the above argument. 
To finish the proof, we require that approximating $Z_G(\lam)$ is hard for $\Delta$-regular \emph{triangle-free} graphs $G$ when $\lam>\lam_c$. This follows directly from the proof of Sly and Sun~\cite{SS14}, as their gadget which shows hardness for $\Delta$-regular graphs contains no triangles. 
Thus, we have the following analogue of Theorem~\ref{thmAPreduction}, reducing the problem $\mathrm{HC}'(\lam,\Delta)$ of computing $Z_G(\lam)$ for a $\Delta$-regular triangle-free graph $G$ to the problem $\mathrm{IS}'(\alpha,\Delta)$ of computing $i_{\lfloor\alpha n\rfloor}(G)$ for a $\Delta$-regular triangle-free graph $G$.

\begin{theorem}
    Given $\delta>0$ there exists $\Delta_0$ such that the following holds for all $\Delta \ge\Delta_0$. For every $\alpha \in ((1+\delta)/\Delta, 1/2)$, there exists $\lam > \lam_c(\Delta)$ and an algorithm such that
    \begin{enumerate}[label=\textup{(\roman*)}]
		\item Given an instance $G'$ of $\mathrm{HC'}(\lam,\Delta)$, the algorithm constructs an instance $G'$ of $\mathrm{IS'}(\alpha,\Delta)$ with size polynomial in $|G|$.
			\item Given a relative $\eps/2$-approximation of $i_{\lfloor\alpha |G'|\rfloor}(G')$, a relative $\eps$-approximation of $Z_G(\lam)$ can be computed in polynomial time.
\end{enumerate}
\end{theorem}

\noindent
Along with Theorem~\ref{thmHChard}, this implies Theorem~\ref{thmTriangleFree}\ref{itmTriangleFreeHardness}.

\section{Anti-ferromagnetic Ising model}\label{secAFIsing}

\subsection{Proof of Theorem~\ref{thmanti-ferro}}

The proof of Theorem~\ref{thmanti-ferro} is essentially the same as the proof of Theorem~\ref{thm:main}, so we describe only the different ingredients needed. 
Our proof of Theorem~\ref{thm:main} relies on a uniqueness/computational threshold $\lam_c$ such that the following hold: 
\begin{enumerate}[label=\textup{(\roman*)}]
    \item\label{itmUniqueness} For $0 < \lam < \lam_c$ we have rapid mixing of Glauber dynamics of the model on every $G\in \cG$ (Theorem~\ref{thmSampleLambda}) and a complex zero-free region $U$ containing the real interval $(0, \lam_c)$ such that for all $G\in \cG$ the partition function $Z_G(\lam)$ is nonzero in $U$ (Theorem~\ref{thm:zerofree}).
    \item\label{itmHardness} For $\lam > \lam_c$ the problem of approximating $Z_G(\lam)$ is hard on $\cG$ (Theorem~\ref{thmHChard}).
\end{enumerate}

Given~\ref{itmUniqueness} and~\ref{itmHardness}, the computational threshold for computing coefficients of the partition function is given by minimizing the occupancy fraction at $\lam_c$ over the class of graphs $\cG$. 

For the case of the anti-ferromagnetic Ising model (Theorem~\ref{thmanti-ferro}), property \ref{itmUniqueness} is provided by the following recent results.

\begin{theorem}[\cite{CLV20,CLV20a}]
\label{thmIsingMix}
Fix $\Delta \ge 3$, $B \in (0, B_c(\Delta))$ and $\lam< \lam_c(\Delta, B)$.  Then the Glauber dynamics for the Ising model with parameters $B, \lam$ mix in time $O(n \log n)$ on the class of graphs of maximum degree $\Delta$. 
\end{theorem}

\begin{theorem}[\cite{SST14,LLY13,SS20b}]
\label{thmIsingZeros}
Fix $\Delta \ge 3$, $B \in (0, B_c(\Delta))$ and $\lam_0< \lam_c(\Delta, B)$.  There exists $\delta>0$ and a simply connected region $U$ in the complex plane containing all points within a distance $\delta$ of the real interval $[0,\lam_0]$ so that $Z_G(\beta, \lam) \ne 0$ for $\lam \in U$ and any graph $G$ of maximum degree $\Delta$. 
\end{theorem}

\noindent
Property~\ref{itmHardness} is given by the following.

\begin{theorem}[\cite{galanis2016inapproximability,SS14}]
\label{thmIsinghardness}
Fix $\Delta \ge 3$, $B \in (0, B_c(\Delta))$, $\lam \in (\lam_c(\Delta, B),1]$, and $\eps >0$.  
Unless NP=RP, there is no polynomial-time algorithm that outputs a relative $\eps$-approximation of $Z_G(B,\lam)$ on the class of $\Delta$-regular graphs. \end{theorem}

\noindent
With these inputs, the proof of Theorem~\ref{thmanti-ferro} is essentially identical to the proof of Theorem~\ref{thm:main}.

For the algorithm, the definition of $\alpha_{\mathrm{inf}}(B, \lam_c,\Delta)$ and the assumption that $\alpha < \alpha_{\mathrm{inf}}(B, \lam_c,\Delta)$ ensures that there exists $\lam < \lam_c(\Delta, B)$ so that the expected number of $+$ spins in a sample from $\mu_{G,B,\lam}$ is $\lfloor \alpha n \rfloor$; since the Glauber dynamics mixes rapidly (Theorem~\ref{thmIsingMix}) we can use a binary search to find an approximate such $\lam$; using the zero-freeness from Theorem~\ref{thmIsingZeros} and  rapid mixing we can show that a sample from $\mu_{G,B,\lam}$ has probability $\Theta\big( \frac{1}{n \log n} \big) $ of achieving exactly the desired number of $+$ spins.  The approximate sampling algorithm is then rejection sampling with approximate samples from $\mu_{G,B,\lam}$ obtained from the Glauber dynamics.

  In the proof of hardness,  the definition of $\alpha_{\mathrm{inf}}(B, \lam_c,\Delta)$ and the assumption that $\alpha > \alpha_{\mathrm{inf}}(B, \lam_c,\Delta)$ ensures that there exists a fixed $\Delta$-regular graph $K$ with $\alpha_K(B,\lam_c) < \alpha$.  Moreover, since $\alpha_K(B,1)=1/2$ (as when $\lam=1$ the model is symmetric under swapping the spins), by continuity there exists $\lam \in ( \lam_c(\Delta,B),1)$ so that $\alpha_K(B,\lam)=\alpha$.  We then reduce the problem of approximating $Z_G(B,\lam)$ for this value of $\lam$ (which by Theorem~\ref{thmIsinghardness} is intractable unless NP=RP)  to that of approximating $c_{\lfloor \alpha n \rfloor}(G',B)$ for the graph $G'$ consisting of the union of $G$ and $r$ copies of the graph $K$, with $r = \Theta(n^2/\eps)$.  To relate $c_{\lfloor \alpha n \rfloor}(G',B)$ to $Z_G(B,\lam)$ we use  the local central limit theorem of  Theorem~\ref{thm:gnedenko} as in the proof of~\eqref{eqReductionEquation} above.

In order to give an explicit value for the computational threshold $\alpha_{\mathrm{inf}}(B, \lam_c,\Delta)$, one must solve the extremal problem. In Conjecture~\ref{conjIsingExtremal} we posit that the answer is given by the complete graph $K_{\Delta +1}$.

\providecommand{\bysame}{\leavevmode\hbox to3em{\hrulefill}\thinspace}
\providecommand{\MR}{\relax\ifhmode\unskip\space\fi MR }
\providecommand{\MRhref}[2]{%
  \href{http://www.ams.org/mathscinet-getitem?mr=#1}{#2}
}
\providecommand{\href}[2]{#2}

\end{document}